\documentclass{llncs}
\usepackage{enumitem}
\setlist{nolistsep}
\usepackage{color}
\usepackage{algorithm}
\usepackage{graphicx}
\usepackage{float}
\usepackage{amsmath, amssymb}
\usepackage{algorithm,algorithmic}
\usepackage{tabularx}

\newcolumntype{R}{>{\raggedleft\arraybackslash}X}
\newcolumntype{L}{>{\raggedright\arraybackslash}X}

\title{Heuristic algorithms for the min-max edge 2-coloring problem}
\author{Radu Stefan Mincu\inst{1} \and Alexandru Popa\inst{1,2,}\thanks{ \scriptsize This work was supported by the research programme PN 1819 ``Advanced IT resources to support digital transformation processes in the economy and society - RESINFO-TD'' (2018), project PN 1819-01-01 ``New research in complex systems modelling and optimization with applications in industry, business and cloud computing'', funded by the Ministry of Research and Innovation.}}
\institute{Department of Computer Science, University of Bucharest
\and 
National Institute for Research and Development in Informatics \\
E-mail: \email{mincu.radu@fmi.unibuc.ro, alexandru.popa@fmi.unibuc.ro} }

\begin{document}
\maketitle

\vspace{-0.3cm}

\begin{abstract}

In multi-channel Wireless Mesh Networks (WMN), each node is able to use multiple non-overlapping frequency channels. Raniwala et al. (MC2R 2004, \mbox{INFOCOM} 2005) propose and study several such architectures in which a computer can have multiple network interface cards. These architectures are modeled as a graph problem named \emph{maximum edge $q$-coloring} and studied in several papers by Feng et. al (TAMC 2007), Adamaszek and Popa (ISAAC 2010, JDA 2016). Later on Larjomaa and Popa (IWOCA 2014, JGAA 2015) define and study an alternative variant, named the \emph{min-max edge $q$-coloring}.

The above mentioned graph problems, namely the maximum edge \mbox{$q$-coloring} and the min-max edge \mbox{$q$-coloring} are studied mainly from the theoretical perspective. In this paper, we study the min-max edge \mbox{2-coloring} problem from a practical perspective. More precisely, we introduce, implement and test four heuristic approximation algorithms for the min-max edge $2$-coloring problem. These algorithms are based on a \emph{Breadth First Search} (BFS)-based heuristic and on \emph{local search} methods like basic \emph{hill climbing}, \emph{simulated annealing} and \emph{tabu search} techniques, respectively. Although several algorithms for particular graph classes were proposed by Larjomaa and Popa (e.g., trees, planar graphs, cliques, bi-cliques, hypergraphs), we design the first algorithms for general graphs. 

We study and compare the running data for all algorithms on Unit Disk Graphs, as well as some graphs from the DIMACS vertex coloring benchmark dataset. 

\end{abstract}

\section{Introduction}
\paragraph{\sc Motivation.}

In multi-channel Wireless Mesh Networks (WMN), each node is able to use multiple non-overlapping frequency channels. The use of many channels inside the same network can significantly improve overall performance. Interference from neighboring nodes can be decreased substantially  when nodes do not need to use the same radio channel for every link. Multiple radio channels in the network imply that at least some of the nodes need to handle more than one channel at a time. In many proposed designs the multi-channel feature is achieved by packet-by-packet reconfiguration of the radio \cite{muir1998,kyasanur2005,so2004}. However, one of the drawbacks of this kind of continuous channel switching of a single radio interface is that it requires precise synchronization throughout the network.

An alternative approach would be to fit multiple radio interfaces to each node, thus allowing a more persistent channel allocation per interface. A couple of such multi-NIC (network interface card) architectures have been proposed by Raniwala et al. \cite{RaniwalaGC04,RaniwalaC05}. Their simulation and testbed experiments show a promising improvement with only two NICs per node, compared to a single-channel WMN. Another appealing feature of these architectures is that they are based on readily available, commodity IEEE 802.11 interfaces, requiring only systems software modification. 

The scenario of two or more NICs per node with fixed channels imposes some limitations to the assignment of channels on each interface. In order to set up a link between two nodes, both of them have to have at least one of their interfaces set to the same channel. On the other hand, links inside an interference range should use as many different channels as possible. Thus, the channels need to be assigned carefully in order to both keep every required link possible and maximize useful bandwidth throughout the network.

\vspace{-0.2cm}

\paragraph{\sc Problem definition}

The channel assignment problem can be modeled as a type of edge coloring problem: given a graph $G$, the edges have to be colored so that there are at most $q$ different colors incident to each vertex. Here, vertices, edges and colors represent network nodes, links and channels, respectively. A coloring that satisfies this constraint, is called an \emph{edge $q$-coloring}. Note, that the coloring constraint differs from the traditional coloring problems, where adjacent items are not allowed to have the same color. Also the goal is different; instead of minimizing, we want to maximize the number of different colors in an edge $q$-coloring.

Initially, the channel assignment was formulated as the \emph{maximum edge $q$-coloring problem}, where the goal was to maximize the total number of colors in a $q$-coloring. The drawback of this model is that in an optimal solution the same color is assigned to many edges while other colors are used only once. We remind the reader that in the wireless mesh network setting, having the same color assigned to many edges is equivalent to having the same frequency used many times, and therefore, having interference. Since the goal of the application is to minimize the interference, max edge $q$-coloring is perhaps not the ideal theoretical formulation (although max edge $q$-coloring is still interesting as a combinatorial problem).  Instead, it is more relevant for the network application to try to have the color components as balanced as possible. Thus, the \emph{min-max edge $q$-coloring} had been introduced, where the goal is to minimize the maximum size of a color group. The formal definition of the min-max $q$-coloring follows.

\vspace{-0.1cm}

\begin{problem}[Min-max edge $q$-coloring]
\label{problem:minmax}
Given a graph $G=(V,E)$, find an edge $q$-coloring $\sigma$ of $G$ such that the amount $max_c|\{e \in E | \sigma(e)=c\}|$ is minimized. In other words, find an edge coloring that minimizes the size of the largest set of edges with the same color.
\end{problem}

\paragraph{\sc Previous work.} The problem of finding a maximum edge $q$-coloring of a given graph has been first studied by Feng et al.~\cite{FengZQW07,FengCZ08,FengZW09}. They provide a \mbox{$2$-approximation} algorithm for $q=2$ and a $(1+\frac{4q-2}{3q^2-5q+2})$-approximation for $q > 2$. They show that the problem is solvable in polynomial time for trees and complete graphs in the case $q=2$, but the complexity for general graphs has been left as an open problem. Later, Adamaszek and Popa~\cite{AP10} show that the problem is APX-hard and present a $5/3$-approximation algorithm for graphs which have a perfect matching. The maximum edge $q$-coloring is also considered in combinatorics and is a particular case of the anti-Ramsey number. For a brief description of the connection of the two problems, the reader can refer to~\cite{AP10}.

Larjomaa and Popa \cite{larjomaa2014min,LarjomaaP15}  introduce and study the min-max edge $q$-coloring problem. They prove that the problem is NP-hard for any $q \ge 2$ and show an exact polynomial time algorithm for trees, for $q=2$. Moreover, Larjomaa and Popa~\cite{LarjomaaP15} analyze the value of the optimal solution on special classes of graphs: cliques, bicliques and hypercubes. They provide the exact formulas of the optimal solutions for cliques. For bicliques they present a lower bound which is tight when both parts of the graph have an even number of vertices (and almost tight for the other cases). For a hypergraph $Q_n$ they give a lower bound which is tight for even $n$, and similarly, almost tight for odd $n$. Although these classes of graphs have a very simple structure, finding lower bounds is much more difficult than in the case of the max edge $q$-coloring problem.

A good lower bound of the optimal solution is necessary in order to design approximation algorithms. For the min-max edge $q$-coloring problem, a trivial lower bound is half of the maximum degree. Larjomaa and Popa~\cite{LarjomaaP15} show another lower bound in terms of the average degree of the graph. Larjomaa and Popa~\cite{LarjomaaP15} also present an approximation algorithm for planar graphs which achieves a sublinear approximation ratio. The algorithm uses a theorem of Lipton and Tarjan~\cite{Lipton1980} which says that a planar graph admits a small balanced separator.
   
\vspace{-8pt}

\paragraph*{\sc Our results.}
Although the min-max $q$-coloring problem has been studied for particular classes of graphs, little has been done for general graphs in the sense of an approximation algorithm. 
As such, we design, implement and analyze algorithms for the min-max 2-coloring problem for general graphs.

The paper is organized as follows. In Section \ref{section:bfs} we show a Breadth First Search (BFS)-inspired approach to approximating min-max 2-coloring. In Section~\ref{section:localsearch} we present min-max $q$-coloring as a local search problem  in the context of combinatorial optimization. Subsequently, we build the necessary tools to tackle min-max edge 2-coloring as a local search problem (provide neighborhood structure, auxiliary objective function). After this framework is built, we construct algorithms to solve the problem using hill climbing (its basic nature led to omitting the full algorithm from this paper), simulated annealing (Subsection \ref{section:sa}) and tabu search (Subsection \ref{section:ts}) techniques. Finally, we reveal some experimental results in Section \ref{section:experiments} and provide insight into the difficulty of the problem and the nature of the methods we employ to solve it.
We reveal a simple design for a BFS-inspired algorithm that yields good results while having the benefit of the linear time complexity of BFS.
We provide evidence that all of our \mbox{local} search algorithms successfully exploit the search space gradient in improving their working solutions as shown by a linear decrease in the objective function.
We show that a simple hill climbing approach produces reasonably good solutions using a low number of iterations over the initial solution. Algorithms \ref{alg:annealing} and \ref{alg:tabu} (based on simulated annealing and tabu search techniques) take longer to complete but manage to escape local optima and achieve better solutions.

In the Experimental Results (Section \ref{section:experiments}) we describe the testing dataset, analyze the implementation of the local search algorithms and show the behavior of the described algorithms on our selected dataset. The results are encouraging while considering  the upper bounds for the optimum solutions for a selection of the input graphs that are obtained with an Integer Linear Program (ILP).
%%%%%%%%%%%%%%%

\section{A BFS-inspired heuristic algorithm}
\label{section:bfs}

We show a simple algorithm for approximating the min-max edge \mbox{2-coloring} by using \emph{Breadth First Search} (BFS). The idea is to color the uncolored edges incident to each subsequent ``level'' in a BFS with a distinct color. The ``levels'' denote the starting vertex, then its neighbors, then the neighbors of the neighbors and so on. The full algorithm is presented as Algorithm \ref{alg:bfs}. The algorithm takes time $O(n+m)$, same as BFS. We can improve the base algorithm by coloring disconnected colored components with distinct colors as shown in \mbox{step 5}. By using a disjoint set forest data structure we may quickly determine these disconnected components during the edge coloring step for only a small overhead of $O(\alpha(m))$, $\alpha$ denoting the inverse of the Ackermann function.

\begin{algorithm}[!hbt]
\caption{input: graph $G=(V,E)$, an initial vertex $v_0$}
\label{alg:bfs}
$1:$ Let there be two sets $Q_1 \leftarrow \{v_0\}$ and $Q_2 \leftarrow \emptyset$. Mark $v_0$ as visited. Integer $c \leftarrow 1$.\\
$2:$ Color all uncolored edges $(v_i,v_j)$ incident to each $v_i \in Q_1$ using integer color $c$. \\
$3:$ During the previous operation, add all unvisited $v_j$ (neighbors of $v_i$) to set $Q_2$.\\
$4:$ Mark all these $v_j$ as visited.\\
$5:$ (Improvement step) Consider the subgraph containing all the edges colored with integer $c$. Color each disconnected component in this subgraph with a new color obtained by incrementing $c$.\\
$6:$ Let $c \leftarrow c + 1$, $Q_1 \leftarrow Q_2$ and $Q_2 \leftarrow \emptyset$ \\
$7:$ If $Q_1=\emptyset$ then the algorithm terminates. Else, continue with step $2$. \\

\end{algorithm}

\begin{theorem}
Algorithm \ref{alg:bfs} produces a valid 2-coloring.
\end{theorem}
\begin{proof}
The colored subgraph $G_i$ grows at each iteration $i$ of the algorithm by adding a new layer of previously uncolored edges. The vertices along the border of $G_i$ all have incident edges with the same color. At step $i+1$ these vertices may obtain a second incident color if they had any uncolored incident edges in the main graph at iteration step $i$. Assume that $G_0=(V,\emptyset)$ and at step $i$ there is a 2-coloring using $i$ colors in $G_i$, but in $G_{i+1}$ we add a third incident color different from $c_{i+1}$ to some vertex $p$ (which has to be at the border of $G_i$). This third color comes from an edge that is incident to both $p$ and a vertex $q$ from the border of $G_{i+1}$. This edge can only be colored with $c_{i+1}$, contradiction. \qed
\end{proof}

\section{Local search algorithms}
\label{section:localsearch}

Min-max edge $q$-coloring (including $2$-coloring) can naturally be modeled as a combinatorial optimization problem:
\begin{itemize}
\item \textit{a solution} $\omega$ is a color mapping from the edge set of the graph to a set of positive integers, for example.
\item \textit{the objective function} $f(\omega)$ used to evaluate the quality of the solution is the largest number of edges that share the same color. Our purpose is to minimize this amount, as such, it is a \textit{minimization problem}.
\item \textit{the constraint} is that the set of edges incident to a vertex can contain edges that are colored with at most $q$ (respectively, two for $2$-coloring) different colors.
\item \textit{a feasible solution} will respect the constraint across all vertices while \textit{an unfeasible solution} will not.

\end{itemize}

To solve this problem using local search, there are a few more requirements to fulfill:
\begin{itemize}
\item \textit{some initial solution} $\omega_0$ to start improving upon.
\item \textit{a neighborhood structure} $N(\omega)$ to provide slightly modified candidate colorings that we will evaluate with our objective function. If a neighbor is better in terms of the objective function then we select it as current solution (i.e. $\omega_{current} \leftarrow \omega_{best} \in N(\omega_{current})$).
\item \textit{some stopping criteria} to prevent the algorithm from looping.

\end{itemize}

For our neighborhood structure we choose \textit{operations based neighborhood}, that is to say, we apply some local modifications or \textit{moves} (i.e. color changes) to some components of the current solution (i.e. edges). The set of moves applied to every component of the solution $\omega$ will construct the neighborhood $N(\omega)$.

\textit{Notation 1.} In the following we refer to the \textit{color class of a vertex} as being the set of colors of its incident edges. We use $cc(v)$ to denote the color class of a vertex $v$. By definition, $cc(v)=\bigcup _{(v,v') \in E} \sigma((v,v'))$, where $\sigma$ is an edge coloring.

We now consider a move set that can be applied only on feasible solutions (i.e. $2$-colorings) and will also produce only feasible solutions.

The defined moves can only be applied in certain cases depending on the color classes of the endpoint vertices of the edge we operate on. Such scenarios are depicted in Figure \ref{fig:moves} but do not reveal all possible cases. The omitted cases are those that result in the removal of a color from either or both of the color classes for \textit{exchange}, \textit{connect} and \textit{create}.

The effect of each move in our defined move set is detailed below:

\begin{enumerate}
\item \textit{Exchange}. Applicable iff the color class of either endpoint is included in the other (or equal) and at least one of the endpoint vertices has a color class of cardinality $2$: change the color of the edge to the other color in the endpoints' color classes.

$\forall e = (v,v')$ if $\{col\} = cc(v) \cup cc(v') \setminus \{e_{color}\}$: $e_{color} \leftarrow col$
\item \textit{Connect}. Applicable iff the color classes of the endpoint vertices are  both of cardinality $2$ and not equal: repaint the edge color, as well as all the edges using the other two colors in the respective endpoints' color classes with a new, unified color. 

$\forall e = (v,v')$ if $\{col_1,col_2\} = cc(v) \cup cc(v') \setminus \{e_{color}\}$:

 $e_{color} \leftarrow col_{new}, \forall e' \in \left( \bigcup \limits_{\sigma(e_1)=col_1} e_1 \right) \cup \left( \bigcup \limits_{\sigma(e_2)=col_2} e_2 \right): e'_{color} \leftarrow col_{new}$

\item \textit{Create}. Applicable iff the endpoints both have color classes of cardinality $1$: assign a new color to the edge.

$\forall e = (v,v')$ if $\emptyset = cc(v) \cup cc(v') \setminus \{e_{color}\}$: $e_{color} \leftarrow col_{new}$
\item \textit{Merge}. Essentially an operation that recolors two neighboring colored components with a new color. For consistency it is defined as operating on an edge like the other moves.

\end{enumerate}

\begin{figure}[hbt]
\begin{center}
\begin{tabularx}{\textwidth}{@{}c  @{}c @{}c  @{}c @{} c@{}}
\begin{tabular}{c}
\includegraphics[width=2.33cm]{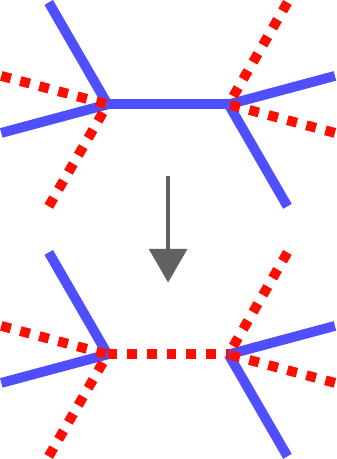}
\end{tabular}
&
 \begin{tabular}{c}
\includegraphics[width=2.33cm]{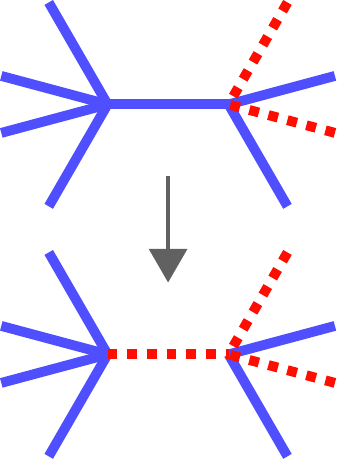}
 \end{tabular}
 &
 \begin{tabular}{c}
\includegraphics[width=2.33cm]{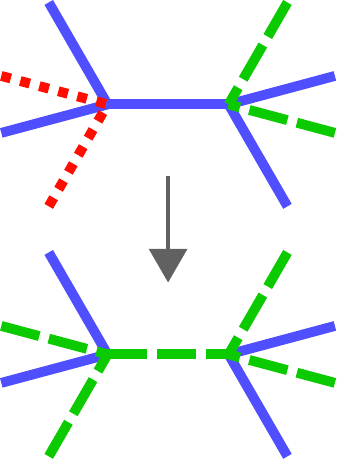}
 \end{tabular}
 &
 \begin{tabular}{c}
\includegraphics[width=2.33cm]{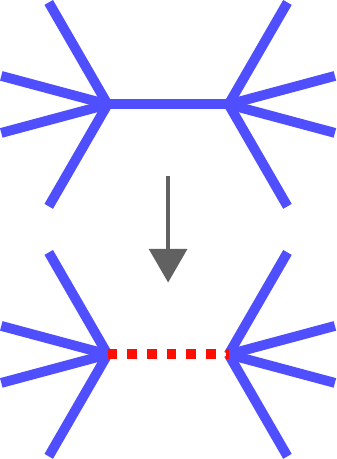}
 \end{tabular}
 &
 \begin{tabular}{c}
\includegraphics[width=2.33cm]{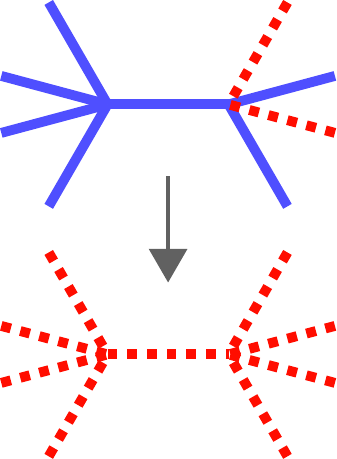}
 \end{tabular}
 \\
 (exchange 1)
 &
 (exchange 2)
 &
 (connect)
 &
 (create)
 &
 (merge) 
 \\
$e_{col} \leftarrow$ $c_{other}$
&
$e_{col} \leftarrow$ $c_{other}$
&
$e_{col} \leftarrow (c_1 \cup c_2)$
&
$e_{col} \leftarrow$ $c_{new}$
&
$e_{col} \cup c_{other}$
 \\
$(\{e_c,c_o\},\{e_c,c_o\})$
&
$(\{e_c\},\{e_c,c_o\})$
&
$(\{e_c,c_1\},\{e_c,c_2\})$
&
$(\{e_c\},\{e_c\})$
&
$(\{e_c\},\{e_c,c_o\})$
 \\
\end{tabularx}
\end{center}
\caption{Illustration of the considered move set in our local search algorithms. The central horizontal edge in all scenarios is the one that considers changing its color ($e_{col}$). The operation may affect the color of edges other than the horizontal one, as in \textit{merge} and \textit{connect}. Here, the $\cup$ operator stands for unifying two colors. The bottom row shows the vertex color classes where the move is applicable.}

\label{fig:moves}

\end{figure}

\begin{theorem}
The move set defined above can only produce $2$-colorings.
\end{theorem}

\begin{proof}

All of the moves change the edge color and never add a third color to the edge endpoint vertices' color classes. We can observe that:
\begin{enumerate}
\item \textit{Exchange case 1} does not add a new color to either color class. At most it can remove one from either or both.
\item \textit{Exchange case 2} can at most add a color to a color class of cardinality $1$.
\item \textit{Connect} modifies the colors in the color classes but they remain of cardinality $2$ (or may decrease). Other affected edges maintain their color class cardinality (or may decrease).
\item \textit{Create} adds a color to color classes of cardinality $1$. A color class may remain of cardinality $1$ if the respective endpoint has degree $1$.
\item \textit{Merge} produces color classes of cardinality $1$. Other affected edges maintain their color class cardinality (or may decrease).
\end{enumerate}
Therefore any move applied on a $2$-coloring will produce  a $2$-coloring.
 \qed
\end{proof}

\textit{Notation 2.} We refer to an edge as being \textit{color critical} if by removing this edge from a subgraph containing all of the edges that share its color will result in the number of connected components increasing in that subgraph.

In our algorithms, recoloring of a color critical edge will cause one of the resulting connected components to be colored with a new color.

Suppose that we have a solution $\omega$ and we operate on an edge which is colored with the most frequent color. Then, the moves defined above will affect $f(\omega)$, our objective function, in the following way:
\begin{enumerate}
\item \textit{Exchange} will produce $f(\omega') \leftarrow f(\omega) - 1$ if the edge is not color critical. Otherwise the objective function can decrease by more than $1$. However, if the other color present in the endpoints' color classes has the same frequency as the one on the edge we operate on, then $f(\omega') \leftarrow f(\omega) + 1$. 
\item \textit{Connect} will produce $f(\omega') \leftarrow f(\omega) - 1$ if the edge is not color critical. Otherwise the objective function can decrease by more than $1$.  However, if the sum of the frequencies of the other two colors in the endpoints' color classes is equal to or exceeds that of the edge we operate on, then $f(\omega')$ will increase by $1$ or more. 
\item \textit{Create} will produce $f(\omega') \leftarrow f(\omega) - 1$ if the edge is not color critical. Otherwise the objective function can decrease by more than $1$.
\item \textit{Merge} will affect $f(\omega')$ if the number of colored edges that use the new color exceeds the previous objective function value. The new $f(\omega')$ can be no more than twice $f(\omega)$, just like with the \textit{connect} case.
\end{enumerate}

The cases when the moves affect a color critical edge of the most frequent color are tricky: to properly calculate the impact on the objective function one needs to perform for example a depth first search across all the neighboring edges to update the objective function. When we want to explore the entire neighborhood of a solution this becomes computationally expensive as we need to perform depth first searches for all edges in $O(|E|(|V|+|E|)$ for all iterations of our local search algorithms.

To avoid this, one can use probabilistic sampling of the neighborhood. In our implementations we prefer to discard this computation entirely, as it is certain that the objective function is decreased by at least $1$ with all of the moves if we are careful about avoiding the special cases that worsen the value.

However, if the objective function can only decrease by $1$ in all cases, there will not be sufficient information to drive the search to good solutions. As such, we use an auxiliary objective function in terms of defining an attractiveness value for each of the moves.

\textit{Notation 3.} In the following we denote $count(c)$ to represent the number of edges that are colored with color $c \in C$. Formally, $count(c) = |\{e \in E | \sigma(e)=c\}|$.

Next, we define the attractiveness for each move (which must be \textit{maximized}):
\begin{enumerate}
\item $att_{exchange}(e,\omega) = b_1 + w_1  \cdot count(e_{color}) \cdot \frac{f(\omega)-count(col_{other})}{f(\omega)}$
\item $att_{connect}(e,\omega) = b_2 +  w_2  \cdot count(e_{color}) \cdot \frac{f(\omega)-count(col_1)-count(col_2)-1}{f(\omega)}$
\item $att_{create}(e,\omega) = b_3 + w_3  \cdot count(e_{color}) \cdot \frac{1}{f(\omega)}$
\item $att_{merge}(e,\omega) = b_4 +  w_4  \cdot  \frac{f(\omega)-count(e_{color})-count(col_{other})}{f(\omega)}$
\end{enumerate}

The constants $b_i$, $w_i$ are used for the fine tuning of the attractiveness values. Observe that for \textit{connect} and \textit{merge} the fraction part of the attractiveness will be $0$ when the newly colored components reach exactly the size of $f(\omega)$ and negative if they exceed $f(\omega)$ (and thus worsen the new objective function value).

With all of the above we have all the elements required to build a simple hill climbing  algorithm to approximate min-max $2$-coloring by choosing the most attractive move at each iteration. We may improve upon this algorithm by using metaheuristics for local search such as simulated annealing and tabu search.

\subsection{Simulated annealing algorithm (Algorithm \ref{alg:annealing}).}
\label{section:sa}

\begin{algorithm}[!hbt]
\caption{input: graph $G=(V,E)$}
\label{alg:annealing}
$1:$ Let the working solution $\omega$ to be some initial $2$-coloring.\\
$2:$ Set up some initial starting temperature of the annealing system: $temp \leftarrow temp_{initial}$\\
$3:$ Initialize $e_{chosen} \leftarrow nil$, $att_{chosen} \leftarrow 0$\\
$4:$ Cycle through edges $e \in E$ :
\begin{itemize}
\item (accept improving moves always:)

if $att(e)>att_{chosen}$: $e_{chosen} \leftarrow e$, $att_{chosen} \leftarrow att(e)$
\item (accept worsening moves with temperature dependent probability:)

else if $uniform(0,1) < exp(\frac{att(e) - att_{chosen}}{temp})$: $e_{chosen} \leftarrow e$, $att_{chosen} \leftarrow att(e)$
\end{itemize}
$5:$ Perform the move on the working solution: $\omega \leftarrow move(e_{chosen},\omega)$\\
$6:$ $temp \leftarrow temperature\_decrease\_schedule(temp)$\\
$7:$ Evaluate stopping criteria. If one of the criteria is met terminate the algorithm and output $\omega$.
Otherwise, continue with step 3.\\
\end{algorithm}

In the simulated annealing setup we select an initial temperature for our system and we may accept worsening moves to our working solution with a probability $p$. This probability is affected by the temperature at a particular iteration step of the algorithm and by the loss in move attractiveness of our worsening operation. Lowering temperature causes $p$ to decrease, while moves with low attractiveness also cause a small probability of acceptance.

\subsection{Tabu search algorithm (Algorithm \ref{alg:tabu}).}
\label{section:ts}

\begin{algorithm}[!hbt]
\caption{input: graph $G=(V,E)$}
\label{alg:tabu}
$1:$ Let the working solution $\omega$ to be some initial $2$-coloring.\\
$2:$ Set up the frequency list to contain $0$ for all edges, set up $TabuList \leftarrow nil$.\\
$3:$ Initialize $e_{chosen} \leftarrow nil$, $att_{chosen} \leftarrow 0$\\
$4:$ Cycle through edges $e \in E, e \notin TabuList $ :

if $att(e) - k \cdot frequency(e)>att_{chosen}$: $e_{chosen} \leftarrow e$, $att_{chosen} \leftarrow att(e)- k \cdot frequency(e)$

$5:$ Perform the move on the working solution: $\omega \leftarrow move(e_{chosen},\omega)$\\
$6:$ $TabuList \leftarrow TabuList \cup \{e\}$, $frequency(e) \leftarrow frequency(e) + 1$ \\
$7:$ Evaluate stopping criteria. If one of the criteria is met terminate the algorithm and output $\omega$.
Otherwise, continue with step 3.\\
\end{algorithm}

To explore  the neighborhood of a solution in a more intelligent way we can employ memory to prevent cycling and drive the search to less explored areas of the search space.
To solve the problem using tabu search techniques we use:
\begin{itemize}
\item a simple tabu list providing short-term memory that disallows a move on any edge recently changed;
\item a frequency list on edge moves providing long-term memory. The frequency of an edge $e$ increases by $1$ each time it is used in an \textit{exchange} operation, and move attractiveness values receive a penalty of $-k \cdot frequency(e)$ for some selected constant $k$.
\end{itemize}

\begin{theorem}
Algorithms \ref{alg:annealing} and \ref{alg:tabu} produce valid $2$-colorings.
\end{theorem}

\begin{proof} The algorithms take a feasible solution and apply a move set that only results in feasible solutions ($2$-colorings).  \qed
\end{proof}

\section{Experimental results}
\label{section:experiments}

\subsection{Testing dataset details}

Our testing dataset includes computer generated Unit Disk Graphs and Quasi-Unit Disk Graphs (prefixed with ``udg'' and ``qudg'', respectively in Table \ref{table:localsearch}) that are traditionally used to model Wireless Mesh Networks: two nodes can communicate only if they are within transmission range of each other. 

In the testing setup, these two aforementioned types of graphs were generated by deploying 100, 500 and 1000 vertices with uniformly distributed coordinates over a square with the side measuring 2500 units. The maximum transmission range parameter is specified as a suffix (e.g. \textsl{udg500.140} is a Unit Disk Graph with 500 vertices and transmission range 140). The algorithms were tested on the largest connected component in each graph. The vertex count, edge count and maximum degree of the test graphs are presented in Table \ref{table:localsearch}. 
The transmission range for the Quasi-Unit Disk Graphs varies uniformly between 50\% and 100\% of the maximum specified range. They are generated with the same random seed as their UDG counterparts so that the layout is identical excepting the absence of some edges from the qUDG cases.

The rest of our testing dataset contains graphs that are not themselves modelling wireless networks, for the sake of a more thorough analysis. These graphs are a part of the dataset for the DIMACS graph vertex coloring benchmarks and their high connectivity proves to be quite a challenge for our local search edge-coloring algorithms.  Note that the following graphs featured in our experimental result showcase are geometric graphs, which are more relevant to the Wireless Mesh Network topology: \textsl{dsjr500.1c}, \textsl{dsjr500.5}, \textsl{r250.5}, \textsl{r1000.1c}.

\subsection{ Algorithms \ref{alg:annealing} and \ref{alg:tabu}  implementation details.}

Our implementations are based on the JGraphT Java Graph Library and are made publicly available  by means of a GitHub repository \cite{JavaMinMax}.

In all our local search algorithms we employ a disjoint set forest structure to keep track of colors when we use the \textit{merge} and \textit{connect} moves. The \textit{create} move draws a new color by incrementing a static counter. Every so often, we renumber the colors because all moves except \textit{create} can cause colors to disappear from the coloring. The vertex color classes are maintained inside a hash map structure. After every iteration it is necessary to perform a depth first search to recolor a potential new connected component that becomes disconnected when an edge changes color. Every iteration takes total time $O(|E|\alpha(|V|))$.

In our \textit{simulated annealing  algorithm} we have selected for our cooling schedule the exponential cooling scheme $T' \leftarrow k T$, with $k < 1$, close to $1$, as first proposed by Kirkpatrik et al. \cite{kirkpatrick1983optimization}.

In our \textit{tabu search algorithm} we use hash map structures to keep track of the tabu moves and quickly determine if a move is tabu or not.

\vspace{-9pt}

\subsection{Running data}

Our experiments for the min-max 2-coloring approximation algorithms are performed on a selection of Unit Disk Graphs with increasing vertex density and transmission range, as well as on DIMACS benchmark graphs. 

\begin{figure}[!h]
		\includegraphics[width=0.99\textwidth]{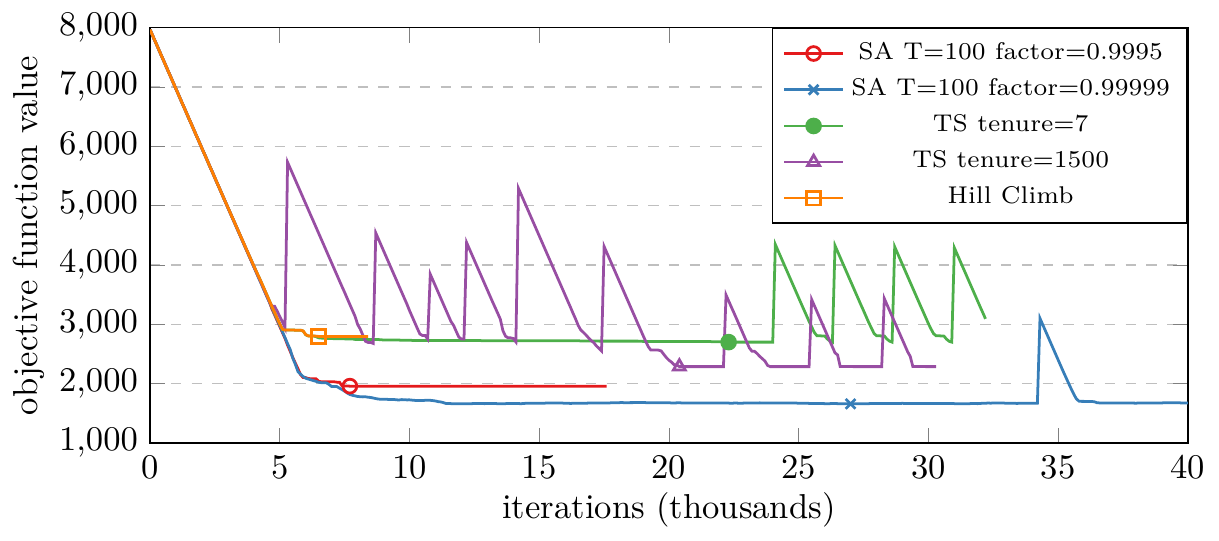}
		\caption{A plot illustrating the quality of the incumbent solution at each iteration of the algorithms on a selected graph (7968 edges). Marks indicate best achieved solutions.} \label{fig:plot}
	\end{figure}

To compare the local search algorithms in terms of the quality of the produced solution the time required to obtain it, we plot the value of the objective function for the current solution at each iteration step in Figure \ref{fig:plot}. To compute the values, we start with an initial solution containing a single colored component and select graph \textsl{qudg1000.220} and the best combination of parameters we have discovered.
The plot makes it easy to observe the linear drop in the objective function for $5000$ units and between iterations $1$ and $5000$. This is a strong point of local search techniques as they exploit the gradient in the search space. Their weakness is that once they reach a local optimum it is hard to escape it as there are no more improving moves to be considered. Simulated annealing approaches this problem by adding randomness to the moves that are selected and we can see the result in the quality of the found solutions. Tabu search will run out of improving moves and will attempt worsening ones to escape the local optimum. The spikes in the objective value function correspond to applying the \textit{merge} move which sometimes almost doubles the last best value. 

For our local search algorithms we choose the initial solution to be either the solution given by Algorithm \ref{alg:bfs} or a single colored component. We stop the algorithms when they fail to produce an improvement for a set number of iterations.

We compare the algorithms with solutions obtained by running an ILP solver (Gurobi) on the the linear program formulation from \cite{larjomaathesis}. To obtain these solutions we had to limit the running time of the solver (3 hours) and the maximum number of allowed colors in the linear program (which in turn decreases the number of variables). As such, the obtained linear program solutions are not the optimum solution and instead are an upper bound for each min-max 2-coloring on the respective graph.

Finally, we present the running data for our algorithms on a selection of graphs in Table \ref{table:localsearch}. The results are encouraging for Unit Disk Graphs and their variants. Our tabu search heuristic applied to the solution of the BFS algorithm seems to consistently yield good results by improving (decreasing) the objective by up to 37\% (21\% on average). For some graphs, the BFS-inspired algorithm seems to create a harder to escape local optimum for the local search heuristic algorithms. This is where simulated annealing produces the best results starting from a blank (single color) initial solution.

\begin{table}[t]
\caption{
Algorithm running data. The first four columns display graph name, vertex count, edge count and maximum vertex degree. The next three columns represent the objective function value obtained by executing hill climbing, simulated annealing and tabu search with a blank (single color) initial solution. The following column is the solution for the BFS-based algorithm and then the solutions for the local search algorithms now starting with it as an initial solution. The last column gives an upper bound for the value of the optimum solution. The best solutions are highlighted. \vspace{-3mm}}
\label{table:localsearch}
\begin{center}
\begin{tabularx}{\textwidth}{ | l |  rRr | RRR | RRRR| r |}
\hline
Graph & $|V|$ & $|E|$ & Deg & HC & SA & TS & BFS & HC$'$ & SA$'$ & TS$'$ & ILP\\
\hline
udg100.400 & 100 & 347 & 12 & 46 & 40 & 32 & 33 & 30 & \textbf{27} & \textbf{27} & 22\\
udg100.600 & 100 & 694 & 22 & 177 & 156 & 131 & 140 & 124 & \textbf{113} & \textbf{113} & 86\\
qudg100.400 & 100 & 232 & 9 & 29 & 19 & 19 & 21 & 19 & \textbf{17} & \textbf{17} & 12\\
qudg100.600 & 100 & 525 & 18 & 104 & 80 & 79 & 88 & 82 & 66 & \textbf{64} & 56\\
udg500.140 & 357 & 893 & 12 & 36 & 33 & 35 & 35 & 32 & \textbf{28} & \textbf{28} & 23\\
udg500.180 & 499 & 1862 & 16 & 198 & 156 & 120 & 77 & 70 & 59 & \textbf{54} & 63\\
udg500.220 & 500 & 2776 & 22 & 840 & 402 & 834 & 195 & 190 & 190 & \textbf{148} & 127\\
udg1000.140 & 1000 & 4641 & 20 & 359 & 173 & 183 & 163 & 141 & 133 & \textbf{102} & -\\
udg1000.180 & 1000 & 7592 & 28 & 2218 & 1073 & 2218 & 579 & 579 & 579 & \textbf{478} & -\\
udg1000.220 & 1000 & 11152 & 39 & 4132 & 3443 & 3854 & \textbf{1058} & \textbf{1058} & \textbf{1058} & \textbf{1058} & -\\
qudg500.140 & 108 & 198 & 9 & 17 & \textbf{12} & 16 & 15 & \textbf{12} & \textbf{12} & \textbf{12} & 10\\
qudg500.180 & 480 & 1281 & 13 & 66 & 52 & 45 & 48 & 34 & 38 & \textbf{31} & 29\\
qudg500.220 & 500 & 1965 & 17 & 251 & 219 & 100 & 90 & 79 & 67 & \textbf{64} & -\\
qudg1000.140 & 998 & 3305 & 14 & 120 & 135 & 78 & 78 & 65 & 60 & \textbf{56} & -\\
qudg1000.180 & 1000 & 5427 & 21 & 614 & 542 & 586 & 295 & 255 & 220 & \textbf{215} & -\\
qudg1000.220 & 1000 & 7968 & 31 & 2601 & 1727 & 2601 & 623 & 619 & 619 & \textbf{614} & -\\
dsjc250.5 & 250 & 15668 & 147 & 7834 & \textbf{7182} & 7193 & 10148 & 7625 & 7625 & 7625 & 5234\\
dsjc500.1 & 500 & 12458 & 68 & 6123 & \textbf{5558} & 6123 & 9162 & 6084 & 6084 & 6084 & 5824\\
dsjc500.5 & 500 & 62624 & 286 & 31115 & \textbf{30089} & 31115 & 42946 & 30766 & 30766 & 33946 & -\\
dsjr500.5 & 500 & 58862 & 388 & 29326 & \textbf{28653} & 29326 & 28724 & 28704 & 28704 & 28704 & -\\
flat300\_28\_0 & 300 & 21695 & 162 & 10586 & 10781 & 10586 & 14604 & \textbf{10551} & \textbf{10551} & \textbf{10551} & -\\
le450\_25c & 450 & 17343 & 179 & 8214 & 8416 & 8214 & 8614 & 8614 & 7549 & \textbf{7286} & 5781\\
le450\_25d & 450 & 17425 & 157 & 8339 & 7763 & 8339 & 8667 & 8667 & 7484 & \textbf{7154} & 5952\\
r250.5 & 250 & 14849 & 191 & 7425 & 6530 & 6806 & 7321 & 7321 & \textbf{5813} & \textbf{5813} & 4950\\
\hline
\end{tabularx}
\end{center}

\vspace{-11pt}
\end{table}

\section{Conclusions and future work}

The newly designed algorithms for the 2-coloring min-max problem offer a practical method of obtaining good solutions without resorting to more time consuming exact methods. 

More techniques to approach the problem may be used, such as recombination heuristics. An idea is to attempt to find some coding for graphs with colored edges suitable for solving $2$-coloring by using a genetic algorithm approach.

It would be interesting to find a constant factor approximation algorithm for min-max edge $q$-coloring.

\bibliographystyle{plain}
\bibliography{bibliography}

\end{document}